\documentclass{amsart}

\usepackage{amsrefs,slashed,booktabs}

\usepackage[T1]{fontenc}

\newtheorem{theorem}{Theorem}[section]
\newtheorem{lemma}[theorem]{Lemma}
\newtheorem{proposition}[theorem]{Proposition}
\newtheorem{corollary}[theorem]{Corollary}

\theoremstyle{definition}
\newtheorem{definition}[theorem]{Definition}
\newtheorem*{defn}{Definition}
\newtheorem{example}[theorem]{Example}

\theoremstyle{remark}
\newtheorem{remark}[theorem]{Remark}

\numberwithin{equation}{section}

\newcommand{\bA}{\mathbb{A}}
\newcommand{\bC}{\mathbb{C}}
\newcommand{\bN}{\mathbb{N}}
\newcommand{\bH}{\mathbb{H}}
\newcommand{\bR}{\mathbb{R}}
\newcommand{\bZ}{\mathbb{Z}}
\newcommand{\bK}{\mathbb{K}}
\newcommand{\cA}{\mathcal{A}}

\newcommand{\cE}{\mathcal{E}}
\newcommand{\cF}{\mathcal{F}}
\newcommand{\cH}{\mathcal{H}}
\newcommand{\cHi}{H^{\infty}}

\newcommand{\cP}{\mathcal{P}}
\newcommand{\cS}{\mathcal{S}}
\newcommand{\cV}{\mathcal{V}}

\newcommand{\ZZ}{\ensuremath{\mathbb{Z}_2}}
\newcommand{\eps}{\varepsilon}
\newcommand{\epsp}{{\varepsilon^\prime}}
\newcommand{\epspp}{{\varepsilon^{\prime\prime}}}

\newcommand{\abs}[1]{\left|{#1}\right|}

\newcommand{\ip}[1]{\left\langle {#1}\right\rangle}
\newcommand{\hp}[1]{\left({#1}\right)}
\newcommand{\ncint}{{\int \!\!\!\!\!\! -}}
\newcommand{\set}[1]{\left\{{#1}\right\}}

\newcommand{\Cstar}{{\ensuremath{C^*}}}

\newcommand{\Star}{\ensuremath{\ast}}
\newcommand{\term}[1]{{\it{#1}\/}}

\DeclareMathOperator{\dom}{Dom}

\DeclareMathOperator{\bCl}{\mathbb{C}l}
\DeclareMathOperator{\bClp}{\mathbb{C}l^{(+)}}
\DeclareMathOperator{\End}{End}

\DeclareMathOperator{\Id}{Id}

\begin{document}

\title[Real almost-commutative spectral triples]{Real structures on almost-commutative spectral triples}

\author{Branimir \'Ca\'ci\'c}
\address{Department of Mathematics\\ California Institute of Technology\\ MC 253-37\\1200 E California Blvd\\ Pasadena, CA 91125}
\email{branimir@caltech.edu}

\subjclass[2010]{Primary 58B34, Secondary 46L87, 81T75}

\keywords{Noncommutative geometry, Spectral triple, Almost-commutative, Dirac-type operator}


\begin{abstract}
We refine the reconstruction theorem for al\-most-com\-mu\-ta\-tive spectral triples to a result for real al\-most-com\-mu\-ta\-tive spectral triples, clarifying in the process both concrete and abstract definitions of real commutative and almost-commutative spectral triples. In particular, we find that a real almost-commutative spectral triple algebraically encodes the commutative \Star-algebra of the base manifold in a canonical way, and that a compact oriented Riemannian manifold admits real (almost-)commutative spectral triples of arbitrary $KO$-dimension. Moreover, we define a notion of smooth family of real finite spectral triples and of the twisting of a concrete real commutative spectral triple by such a family, with interesting $KK$-theoretic and gauge-theoretic implications.
\end{abstract}

\maketitle

The famed Gel'fand--Na{\u\i}mark duality allows for a contravariant functorial identification of the theory of \Cstar-algebras as a theory of noncommutative topological spaces. Analogously, Connes's reconstruction theorem for commutative spectral triples~\cite{Con08} suggests a partial identification, at least at the level of objects, of the theory of spectral triples as a theory of noncommutative manifolds. However, since there is no canonical choice of commutative spectral triple for a compact oriented manifold, it has become traditional in the noncommutative-geometric literature to restrict attention to the case of compact spin manifolds, which admit a canonical Dirac-type operator, \emph{the} Dirac operator, and hence a canonical commutative spectral triple. Indeed, the influence of this example has been profound on the development of the theory of spectral triples, both implicitly through the traditional forms of the definitions of commutative and almost-commutative spectral triples, and explicitly through the focus on \emph{real} spectral triples; the ubiquity of real spectral triples in the literature, together with the explicit use of real spectral triples in applications to theoretical high energy physics, would seem to justify this restriction.

There is, however, another way to approach this issue: every compact oriented manifold admits a Riemannian metric, and a compact oriented Riemannian manifold $X$ admits a canonical spectral triple, namely, the \emph{Hodge--de Rham} spectral triple $(C^\infty(X),L^2(X,\wedge T^\ast_\bC X),d+d^\ast)$. This, then, suggests that the theory of spectral triples might be fruitfully viewed as a theory of noncommutative Riemannian manifolds. Indeed, Lord--Rennie--V\'arilly have developed a full noncommutative generalisation of the Hodge--de Rham spectral triple \emph{qua} candidate notion of noncommutative Riemannian manifold~\cite{LRV}. On the other hand, in the (almost-)com\-mutative context, one might observe that the Hodge--de Rham spectral triple of a compact oriented Riemannian manifold $X$, or more generally, the spectral triple $(C^\infty(X),L^2(X,\cE),D)$ of a Dirac-type operator $D$ acting on a Hermitian vector bundle $\cE \to X$, should constitute a commutative spectral triple in the context of Connes's reconstruction theorem, even if it might only satisfy a weakened version of the orientability condition. This observation formed the basis of our earlier work on almost-commutative spectral triples~\cite{Ca12}, where we obtained a reconstruction theorem for a more general, manifestly global-analytic notion of almost-commutative spectral triple based on general Dirac-type operators; in particular, we were able to refine Connes's reconstruction theorem into a precise noncommutative-geometric characterisation of Dirac-type operators on compact oriented Riemannian manifolds.

In this paper, we refine our earlier definitions and results concerning commutative and almost-commutative spectral triples to accommodate real structures, yielding a reconstruction theorem for real almost-commutative spectral triples. The brunt of the work is in finding the correct concrete (\emph{viz}, global-analytic) and abstract (\emph{viz}, noncommutative-geometric) definitions of real commutative spectral triples and real almost-commutative spectral triples; in particular, we find that the correct abstract definition of real almost-commutative spectral triple reads as follows.

\begin{defn}
An \emph{real almost-commutative spectral triple} of $KO$-dimension $n \bmod 8$ and metric dimension $p$ is a real spectral triple $(A,H,D,J)$ of $KO$-dimension $n \bmod 8$ such that $(A,H,D)$ is a $p$-dimensional (abstract) almost-commutative spectral triple with base $\tilde{A}_J := \set{a \in A \mid Ja^\ast J^\ast = a}$.
\end{defn}

\noindent This definition encapsulates the important general fact, observed to date in specific examples, that a concrete real almost-commutative spectral triple algebraically encodes in a canonical way the commutative \Star-algebra of the base manifold; by contrast, without such additional structure, an abstract almost-commutative spectral triple must be specified with an explicit choice of base. Moreover, we show that a compact oriented Riemannian manifold admits real (almost-)commutative spectral triples of arbitrary $KO$-dimension, providing final confirmation of the essential independence of metric and $KO$-dimensions for real spectral triples.

In addition to these general considerations, we develop a more conservative generalisation of the traditional construction of real almost-commutative spectral triples, defining a notion of \emph{real family}, \emph{viz}, a smooth family of real finite spectral triples with compatible connection, and then defining the product or \emph{twisting} of a real concrete commutative spectral triple by such a real family according to the general framework of D\k{a}browski--Dossena for the product of real spectral triples~\cite{DaDo}. This construction turns out to have a natural interpretation from the standpoint of Bram Mesland's $KK$-theoretic category of spectral triples~\cite{Me09b}, and to have a particularly nice structure with regard to inner fluctuations of the metric.

The author would like to thank his advisor, Matilde Marcolli, and Susama Agarwala, Alan Lai, and Kevin Teh, for helpful comments and conversations, as well as the Erwin Schr\"odinger Institute for its hospitality and its financial and administrative support in the context of the thematic programme ``K-Theory and Quantum Fields.'' The author was also partially supported by NSF grant DMS-1007207.

\section{Preliminaries}

We begin by reviewing the theory of commutative and almost-commutative spectral tripes, including the construction of products of spectral triples.

\subsection{Commutative spectral triples}

Recall that a \emph{Dirac-type operator} on a compact oriented Riemannian manifold $X$ is a first order differential operator $D$ on some Hermitian vector bundle $\cE \to X$, such that $D^2$ is a Laplace-type operator, or equivalently,
\begin{equation}\label{dirac}
	[D,f]^2 = - g(df,df), \quad f \in C^\infty(X);
\end{equation}
if $\cE$ is \ZZ-graded, then we require $D$ to be an odd operator. In particular, then, a Dirac-type operator $D$ on a Hermitian vector bundle $\cE \to X$ induces a Clifford action by
\begin{equation}\label{dirac2}
	c(df) := [D,f], \quad f \in C^\infty(X).
\end{equation}
If $\cE$ is already a Clifford module with Clifford action $c : T^\ast X \to \End(\cE)$, then we require Dirac-type operators on $\cE$ to satisfy Eq.~\ref{dirac2} for the pre-existing Clifford action $c$. Finally, we shall require Dirac-type operators $D$ on $\cE \to X$ to be symmetric, so that they are essentially self-adjoint on $L^2(X,\cE)$ with smooth core $C^\infty(X,\cE)$.

A Dirac-type operator $D$ on $\cE \to X$ gives rise, in the obvious way, to a spectral triple $(C^\infty(X),L^2(X,\cE),D)$, whose basic features are encapsulated in the following definition, slightly modified from Connes's original definition \cites{Con96,Con08}:

\begin{definition}
Let $(A,H,D)$ be a regular spectral triple of metric dimension $p \in \bN$, such that $\cA$ is commutative. We call $(A,H,D)$ a \term{($p$-dimensional) commutative spectral triple} if the following conditions hold:
\begin{enumerate}
	\item {\bf Order one}: For any $a$, $b \in A$, $[[D,a],b] = 0$.
	\item {\bf Finiteness}: One has that $\cHi := \cap_{m} \dom D^{m}$ is a finitely generated projective $A$-module.
	\item {\bf Strong regularity}: One has that $\End_A(\cHi) \subset \cap_k \dom \delta^k$ for the derivation $\delta : T \mapsto [\abs{D},T]$ on $B(H)$.
	\item {\bf Orientability}: There exists an antisymmetric Hoch\-schild $p$-cycle $c \in Z_{p}(A,A)$ such that $\chi = \pi_D(c)$ is a self-adjoint unitary on $\cH$ satisfying $a\chi = \chi a$ and $[D,a]\chi = (-1)^{p+1} \chi[D,a]$ for all $a \in A$.
	\item {\bf Absolute continuity}: The $A$-module $\cHi$ admits a Hermitian structure $\hp{\cdot,\cdot}$ defined by the equality $\ip{\xi,a\eta} = \ncint a\hp{\xi,\eta}\abs{D}^{-p}$ for $a \in A$, $\xi$, $\eta \in \cHi$.
\end{enumerate}
Moreover, we call $(A,H,D)$ \emph{strongly orientable} if $\chi D + D \chi = 0$ when $p$ is even, and $\chi = 1$ when $p$ is odd, and we call $(A,H,D)$ \emph{Dirac-type} if $[D,a]^2 \in A$ for all $a \in A$.
\end{definition}

Thus, for $X$ a compact oriented Riemannian $p$-manifold, a Dirac-type operator $D$ on a Hermitian vector bundle $\cE \to X$ immediately gives rise to a \emph{concrete} $p$-dimensional Dirac-type commutative spectral triple $(C^\infty(X),L^2(X,\cE),D)$; in particular, then, $X$ admits a canonical concrete even $p$-dimensional Dirac-type commutative spectral triple, namely the Hodge--de Rham spectral triple \[(C^\infty(X),L^2(X,\wedge T^\ast_\bC X),d+d^\ast,(-1)^{\abs{\cdot}}),\] where $(-1)^{\abs{\cdot}}$ denotes the \ZZ-grading on $\wedge T_\bC^\ast X$ by parity of the degree.

\begin{remark}
In \cite{Ca12}, we called condition (4) ``weak orientability,'' reserving ``orientability'' for what we call here ``strong orientability.''
\end{remark}

\begin{remark}
What we call ``strong orientability'' is the orientability condition given in the literature, which models commutative spectral triples specifically on the Dirac operator $\slashed{D}$ of a compact spin manifold; the orientability and Dirac-type conditions above were first proposed in \cite{Ca12}, in order to accommodate correctly general Dirac-type operators on compact oriented Riemannian manifolds. In particular, the Dirac-type condition is a straightforward generalisation of Equation \ref{dirac}.
\end{remark}

Already in 1996, Connes conjectured~\cite{Con96} that one could recover a commutative manifold from a commutative spectral triple, just as one can recover a topological space from a commutative \Cstar-algebra \emph{via} Gel'fand--Na{\u\i}mark. Connes finally proved his conjecture, now called the \emph{reconstruction theorem} for commutative spectral triples, in 2008 \cite{Con08}, following a substantial attempt by Rennie--V\'arilly in 2006 \cite{RV06}:

\begin{theorem}[Connes \cite{Con08}*{Thm.\ 1.1}]\label{recon}
Let $(A,H,D)$ be a strongly orientable $p$-dimensional commutative spectral triple. Then there exists a compact oriented manifold $X$ such that $A \cong C^\infty(X)$.
\end{theorem}

Once one has reconstructed the manifold itself, one can proceed to reconstruct the Hermitian vector bundle too, and realise the operator $D$ as an elliptic first-order differential operator, if not necessarily a Dirac-type operator:

\begin{theorem}[Connes \cite{Con96}, Gracia-Bond{\'\i}a--V\'arilly--Figueroa \cite{GBVF}*{Thm.\ 11.2}]\label{babyrecon}
Let $(A,H,D)$ be a strongly orientable $p$-dimensional commutative spectral triple with $A \cong C^\infty(X)$ for some compact orientable manifold $X$. Then there exists a Hermitian vector bundle $\cE \to X$ such that $(A,H,D) \cong (C^\infty(X),L^2(X,\cE),D)$, where $D$ is identified with an essentially self-adjoint elliptic first-order differential operator on $\cE$. 
\end{theorem}

In fact, this last result was stated and proved in the context of reconstructing spin manifolds with spin Dirac operators. The following result of Connes's gives an concise characterisation of spin$^\bC$ manifolds and spin$^\bC$ Dirac operators, possibly with torsion; we shall later recall Plymen's characterisation of spin manifolds amongst spin$^\bC$ manifolds and the resulting theory of real spectral triples.

\begin{corollary}[Connes \cite{Con08}*{Thm.\ 1.2}]\label{spincrecon}
If, moreover, $A^{\prime\prime}$ acts on $H$ with multiplicity $2^{\lfloor p/2 \rfloor}$ (but without needing to assume strong regularity), then $X$ is spin$^\bC$, $\cE \to X$ is a spinor bundle, and $D$ is a Dirac-type operator.
\end{corollary}

Much more generally, by dropping the strong orientability hypothesis from Theorems~\ref{recon} and \ref{babyrecon} \cite{Ca12}*{Proof of Cor.\ 2.19} and exploiting the Dirac-type condition, we obtained a characterisation of compact oriented Riemannian manifolds and Dirac-type operators:

\begin{theorem}[\cite{Ca12}*{Cor.\ 2.19}]\label{rierecon}
Let $(A,H,D)$ be a $p$-dimensional commutative spectral triple. If $(A,H,D)$ is Dirac-type, then there exist a compact oriented Riemannian $p$-manifold $X$ and a Hermitian vector bundle $\cE \to X$ such that $(A,H,D) \cong (C^\infty(X),L^2(X,\cE),D)$, where $D$ is identified with an essentially self-adjoint Dirac-type operator on $\cE$.
\end{theorem}

Thus, a spectral triple is Dirac-type commutative if and only if it is unitarily equivalent to a concrete Dirac-type commutative spectral triple.

\begin{remark}
Here, as in much of the noncommutative-geometric literature, the term ``vector bundle'' is applied to the more general case where the rank is only locally constant, so that the Serre--Swan theorem holds even in the case of a disconnected compact manifold (or, more generally, compact Hausdorff space). However, if one wishes to be scrupulous about using only vector bundles of constant rank, one can use Connes's observation~\cite{Con08}*{Proof of Thm.\ 11.15} that a commutative spectral triple $(A,H,D)$ arises from a Hermitian vector bundle of constant rank $r$ if and only if $A^{\prime\prime}$ acts on $H$ with multiplicity $r$.
\end{remark}

\subsection{Almost-commutative spectral triples}

Let us now turn to the theory of almost-commutative spectral triples. Since such spectral triples were originally defined as the product of a commutative spectral triple with a finite spectral triple, let us first recall the construction of products of spectral triples:

\begin{definition}\label{products}
For $i=1,2$, let $X_i = (A_i,H_i,D_i)$ be a spectral triple. Then the \emph{product} $X_1 \times X_2$ of $X_1$ and $X_2$ is the spectral triple defined as follows:
\begin{enumerate}
	\item If $X_1$ and $X_2$ are both even with \ZZ-gradings $\gamma_1$ and $\gamma_2$ respectively, then
	\[
		X_1 \times X_2 := (A_1 \otimes A_2, H_1 \otimes H_2, D_1 \otimes 1 + \gamma_1 \otimes D_2,\gamma_1 \otimes \gamma_2).
	\]
	\item If $X_1$ is even with \ZZ-grading $\gamma_1$ and $\gamma_2$ is odd, then
	\[
		X_1 \times X_2 := (A_1 \otimes A_2, H_1 \otimes H_2, D_1 \otimes 1 + \gamma_1 \otimes D_2).
	\]
	\item If $X_1$ is odd and $X_2$ is even with \ZZ-grading $\gamma_2$, then
	\[
		X_1 \times X_2 := (A_1 \otimes A_2, H_1 \otimes H_2, D_1 \otimes \gamma_2 + 1 \otimes D_2).
	\]
	\item If $X_1$ and $X_2$ are odd, then
	\[
		X_1 \times X_2 := (A_1 \otimes A_2, H_1 \otimes H_2 \otimes \bC^2, D_1 \otimes 1 \otimes \sigma_1 + 1 \otimes D_2 \otimes \sigma_2, 1 \otimes 1 \otimes \sigma_3),
	\]
	where the $\sigma_k$ are the Pauli sigma matrices:
	\[
		\sigma_0 = \begin{pmatrix} 1 & 0 \\ 0 & 1 \end{pmatrix}, \quad \sigma_1 = \begin{pmatrix} 0 & 1 \\ 1 & 0 \end{pmatrix}, \quad \sigma_2 = \begin{pmatrix} 0 & -i \\ i & 0 \end{pmatrix}, \quad \sigma_3 = \begin{pmatrix} 1 & 0 \\ 0 & -1 \end{pmatrix}.
	\]
\end{enumerate}
\end{definition}

\begin{remark}
In the case where both $X_1$ and $X_2$ are even, one could alternatively construct the Dirac operator of $X_1 \times X_2$ as $D_1 \otimes \gamma_2 + 1 \otimes D_2$; the resulting spectral triple is then unitarily equivalent to $X_1 \times X_2$ as constructed above.
\end{remark}

That the product of spectral triples is indeed a spectral triple does require verification---see \cites{Otgo09,DaDo} for details. In particular, Otgonbayar proves that the product of regular spectral triples is again regular \cite{Otgo09}*{Prop.\ 3.1.32}. 

The conventional definition of almost-commutative spectral triple then reads as follows:

\begin{definition}
A \emph{Cartesian almost-commutative spectral triple} is a spectral triple of the form $X \times F$, where $X$ is a compact spin manifold with fixed spin structure, identified, by abuse of notation, with its canonical commutative spectral triple $(C^\infty(X),L^2(X,\cS),\slashed{D})$, and $F$ is a finite spectral triple.
\end{definition}

\begin{remark}
If $A_F$ is a real \Cstar-algebra, then one should replace $C^\infty(X)$ with $C^\infty(X,\bR)$ in the spectral triple of $X$.
\end{remark}

We have already argued in \cite{Ca12} for a more general, indeed, manifestly global analytic notion of almost-commutative spectral triple, which does not require the base manifold to be spin, accommodates ``non-trivial fibrations'' already present in the literature, and is stable under inner fluctuation of the metric. To write down this definition succinctly, it will be convenient to give the following definitions:

\begin{definition}
Let $X$ be a compact manifold. We define a \emph{bundle of algebras} to be a locally trivial bundle of finite-dimensional \Cstar-algebras. We also define \emph{representation} of a bundle of algebras $\cA \to X$ on a Hermitian vector bundle $\cE \to X$ to be an injective morphism $\cA \to \End(\cE)$ of locally-trivial bundles of finite-dimensional \Cstar-algebras, in which case we call $\cE \to X$ an \emph{$\cA$-module}.
\end{definition}

\begin{remark}
The refinement of Serre--Swan applicable to what we call bundles of algebras, due to Boeijink--van Suijlekom \cite{BvS10}*{Thm.\ 3.8}, actually requires the slightly weaker notion of \emph{algebra bundle}, but the algebra bundles arising in this context are necessarily sub-bundles of a bundle of algebras, namely, the endomorphism bundle of some vector bundle, and are thus bundles of algebras.
\end{remark}

\begin{remark}
For simplicity of notation later on, we shall require bundles of algebra, without further qualification, to be locally trivial bundles of finite-dimensional \emph{complex} \Cstar-algebras. In physical applications, however, one needs \emph{real} bundles of algebras, that is, locally trivial bundles of finite-dimensional real \Cstar-algebras. All the definitions and results in this account hold equally well when real bundles of algebras are used, and we shall continue to remark on any differences from the complex case as they arise.
\end{remark}

\begin{definition}
Let $X$ be a compact oriented Riemannian manifold, and let $\cA \to X$ be a bundle of algebras. We define a \emph{Clifford $\cA$-module} to be a Clifford module $\cE \to X$ together a faithful \Star-representation of $\cA$ commuting with the Clifford action; if $\cE$ is \ZZ-graded, we require sections of $\cA$ to act as even operators on $\cE$.
\end{definition}

Our generalised (concrete) definition is therefore as follows:

\begin{definition}[\cite{Ca12}*{Def.\ 2.3}]
A \emph{concrete almost-commutative spectral triple} is a spectral triple of the form $(C^\infty(X,\cA),L^2(X,\cE),D)$, where $X$ is a compact oriented Riemannian manifold, $\cA \to X$ is a bundle of algebras, $\cE \to X$ is a Clifford $\cA$-module, and $D$ is a Dirac-type operator on $\cE$.
\end{definition}

This new concrete definition then lends itself to the following abstract analogue:

\begin{definition}[\cite{Ca12}*{Def.\ 2.16}]
Let $(A,H,D)$ be spectral triple, which may or may not be even, and let $B$ be a central unital \Star-subalgebra of $A$. We call $(A,H,D)$ a $p$-dimensional \emph{almost-commutative spectral triple} with \emph{base $B$} if the following conditions hold:
\begin{enumerate}
	\item $(B,H,D)$ is a Dirac-type $p$-dimensional commutative spectral triple;
	\item $A$ is a finitely generated projective unital $B$-module-\Star-subalgebra of $\End_B(\cHi)$, where $\cHi := \cap_k \dom D^k$;
	\item $[[D,b],a] = 0$ for all $a \in A$, $b \in B$.
\end{enumerate}
\end{definition}

In particular, a concrete almost-commutative spectral triple with base a compact oriented Riemannian $p$-manifold $X$ is a $p$-dimensional (abstract) almost-commutative spectral triple with base $C^\infty(X)$.

Given these definitions, we obtained a reconstruction theorem for almost-com\-mu\-ta\-tive spectral triples as a simple consequence of Thm.\ \ref{recon} \emph{via} Thm.\ \ref{rierecon}, together with the aforementioned refinement of Serre--Swan for algebra bundles \cite{BvS10}*{Thm.\ 3.8}.

\begin{theorem}[\cite{Ca12}*{Thm.\ 2.17}]\label{acrecon}
Let $(A,H,D)$ be a $p$-dimensional almost-com\-mu\-ta\-tive spectral triple with base $B$. Then there exist a compact oriented Riemannian $p$-manifold $X$, a bundle of algebras $\cA \to X$, and a Clifford $\cA$-module $\cE \to X$, such that $B \cong C^\infty(X)$ and $(A,H,D) \cong (C^\infty(X,\cA),L^2(X,\cE),D)$, where $D$ is identified with an essentially self-adjoint Dirac-type operator on $\cE$.
\end{theorem}

\begin{remark}
When $B$ and $A$ are real \Star-algebras, or if $\cA \to X$ is a real bundle of algebras, then one should replace $C^\infty(X)$ with $C^\infty(X,\bR)$. Note that, on the one hand, if $\cA \to X$ is a real bundle of algebras, \emph{i.e.}, a locally trivial bundle of finite-dimensional real \Cstar-algebras, then $C^\infty(X,\cA)$ will, in particular, be a real pre-\Cstar-algebra, whilst on the other, if $A$ is the algebra of an almost-commutative spectral triple with base $B$, then Boeijink--van Suijlekom's Serre--Swan theorem for algebra bundles, adapted to the real case, realises it as the algebra of sections of a sub-algebra bundle of the real endomorphism bundle $\End_{\bR}(\cE)$ of some Hermitian vector bundle $\cE$, and thus as a real pre-\Cstar-subalgebra of the real pre-\Cstar-algebra of sections of $\End_{\bR}(\cE)$.
\end{remark}

\section{Real structures on (almost-)commutative spectral triples}

We now consider real structures on commutative and almost-commutative spectral triples, in light of their respective reconstruction theorems.

\subsection{Spin structures and real structures}\label{plymen}

Let us begin by recalling the dif\-fer\-en\-tial-geometric motivation for the notion of real structures on spectral triples. In doing so, we shall also motivate the ``exotic'' $KO$-dimensions appearing in the framework of D\k{a}browski--Dossena for products of real spectral triples.

Recall that if $X$ is a compact oriented Riemannian manifold, we may define a finite rank Azumaya bundle $\bClp(X) \to X$ by
\[
	\bClp(X) := \begin{cases} \bCl(X), &\text{if $\dim X$ is even,}\\ \bCl^+(X), &\text{if $\dim X$ is odd.} \end{cases}
\]
The bundle $\bClp(X)$ admits a canonical $\bC$-linear anti-involution $\tau$ defined by
\[
	\tau(\xi_1 \cdots \xi_m) := (-1)^m \xi_m \cdots \xi_1, \quad m \in \begin{cases} \bN, &\text{if $\dim X$ is even,}\\ 2\bN, &\text{if $\dim X$ is odd,} \end{cases}  \quad \xi_i \in \Omega^1(X),
\]
so that if $\cE$ is a $\bClp(X)$-module with $\bClp(X)$-action denoted by \[c : \bClp(X) \to \End(\cE),\] then the dual bundle $\cE^\vee$ is also a $\bClp(X)$-module with $\bClp(X)$-action given by
\[
	c^\vee(\omega) := c(\tau(\omega))^T, \quad \omega \in C^\infty(X,\bClp(X)).
\]
This gives rise to the following noncommutative-geometric characterisation of spin$^\bC$ and spin manifolds, here translated into differential-geometric language:

\begin{theorem}[Plymen \cite{Pl86}*{\S 2}]
Let $X$ be compact oriented Riemannian.
\begin{enumerate}
	\item $X$ is spin$^\bC$ if and only if there exists an irreducible $\bClp(X)$-module.
	\item $X$ is spin if and only if there exists an irreducible $\bClp(X)$-module $\cS$ such that $\cS \cong \cS^\vee$ as $\bClp(X)$-modules.
\end{enumerate}
\end{theorem}

Since $\bClp(X)$ has rank $2^{\lfloor p/2\rfloor}$ when $X$ is $p$-dimensional, part (1) of the above theorem is what allows one to obtain Cor.~\ref{spincrecon} from Thms.~\ref{recon} and \ref{babyrecon}.

Now, for $\epsp = \pm 1$, define a $\bC$-linear anti-involution $\tau_\epsp$ on $\bCl(X)$ by \[\tau_\epsp|_{\Omega^1(X)} = -\epsp \Id_{\Omega^1(X)};\]
by construction, $\tau_{\epsp}$ defines an extension to $\bCl(X)$ of $\tau$ on $\bClp(X)$. Since for a Hermitian vector bundle $\cE$, the dual bundle $\cE^\vee$ is canonically isomorphic to the conjugate bundle $\overline{\cE}$, it is traditional in the noncommutative-geometric literature to reformulate Plymen's characterisation of spin manifolds as follows:

\begin{corollary}[cf.~\cite{GBVF}*{Thms.~9.6, 9.20}]\label{plymenspin}
Let $X$ be a compact oriented Riemannian $n$-manifold. Then $X$ is spin if and only if there exists an irreducible Clifford module $\cS$ together with an antiunitary bundle endomorphism $C$ on $\cS$ satisfying
\begin{enumerate}
	\item $C^2 = \eps \Id_\cS$,
	\item $C c(\omega^\ast) C^\ast = c(\tau_\epsp(\omega))$ for all $\omega \in C^\infty(X,\bCl(X))$,
	\item $C \chi = \epspp \chi C$ for $\chi \in C^\infty(X,\bCl(X))$ the chirality element, when $n$ is even,
\end{enumerate}
where $(\eps,\epsp,\epspp) := (\eps(n),\epsp(n),\epspp(n)) \in \set{\pm 1}^3$ are determined by $n \bmod 8$ as follows (with $\epspp \equiv 1$ is suppressed for $n$ odd):
\begin{equation}\label{kotable}
	\begin{array}{ccccccccc}
	\toprule
	n & 0 & 1 & 2 & 3 & 4 & 5 & 6 & 7\\
	\midrule
	\eps(n) & + & + & - & - & - & - & + & + \\
	\epsp(n) & + & - & + & + & + & - & + & + \\
	\epspp(n) & + & & - & & + & & - & \\
	\bottomrule
	\end{array}
\end{equation}
\end{corollary}

The above folkloric result is the origin of Connes's notion of real structures on spectral triples, and in particular, the above table is the origin of the notion of the $KO$-dimension of a real spectral triple.

\begin{remark}
Condition (2) in the above result can be viewed as specifying the compatibility of $C$ with the Clifford action on $\cS$, for $C$, \emph{a priori}, defines a $\bC$-linear anti-involution $T \mapsto C T^\ast C^\ast$ on $\End(\cS)$.
\end{remark}

\begin{remark}
Suppose that $X$ is spin, and that $\cS$ and $C$ are as above. Then $\slashed{D}C = \epsp C \slashed{D}$ for $\slashed{D}$ the Dirac operator on $\cS$.
\end{remark}

Finally, suppose that $X$ is a compact spin $n$-manifold for $n$ even, and that $\cS$ and $C$ are as given in the above corollary; in particular, we necessarily have that $\epsp = 1$. Let $C_- = C \chi$. Then $C_-$ is an antiunitary bundle automorphism on $\cS$ satisfying
\begin{enumerate}
	\item $C_-^2 = \eps_- \Id_\cS$,
	\item $C_- c(\omega^\ast) C_-^\ast = c(\tau_\epsp(\omega))$ for all $\omega \in C^\infty(X,\bCl(X))$,
	\item $C_- \chi = \epspp_- C_- \chi$, when $n$ is even,
\end{enumerate}
for $(\eps_-,\epsp_-,\epspp_-) := (\eps\epspp,-1,\epspp)$. Thus, as D\k{a}browski--Dossena first observed, one could readily expand the above table to
\begin{equation}\label{newkotable}
	\begin{array}{ccccccccccccc}
	\toprule
	n & 0_+ & 0_- & 1 & 2_+ & 2_- & 3 & 4_+ & 4_- & 5 & 6_+ & 6_- & 7 \\
	\midrule
	\eps(n) 	& + & + & + & - & + & - & - & - & - & + & - & +\\
	\epsp(n) & + & - & - & + & - & + & + & - & - & + & - & +\\
	\epspp(n) & + & + & & - & - & & + & + & & - & - & \\
	\bottomrule
	\end{array}
\end{equation}
where for $n$ even, $n_+$ and $n_-$ denote the two (interchangeable!) possibilities, namely $n_+$ the ``conventional'' $KO$-dimension and $n_-$ the new ``exotic'' $KO$-dimension. Since replacing $C$ with $C\chi$ takes us reversibly between $n_+$ and $n_-$~\cite{DaDo}*{\S 2.3}, the ``exotic'' $KO$-dimensions would seem to offer nothing more than additional notational flexibility. However, as D\k{a}browski--Dossena show (\emph{v.\ infra}), we will need to consider both possibilities simultaneously in order to define consistently products of real spectral triples.

\subsection{Real spectral triples}

Keeping in mind the example of the spinor bundle with Dirac operator and charge conjugation operator over a compact spin manifold, the discussion above generalises even further to the noncommutative setting of spectral triples:

\begin{definition}
A \emph{real spectral triple of $KO$-dimension $n \bmod 8$} is a spectral triple $(A,H,D)$, even with \ZZ-grading $\gamma$ if $n$ is even, together with an antiunitary $J$ on $H$ satisfying:
\begin{enumerate}
	\item $J^2 = \eps \Id_H$,
	\item $DJ = \epsp JD$,
	\item $J \gamma = \epspp \gamma J$ (if $n$ is even),
\end{enumerate}
for $(\eps,\epsp,\epspp) := (\eps(n),\epsp(n),\epspp(n)) \in \set{\pm 1}^3$ depending on $n \bmod 8$ according to Table~\ref{kotable}.
\end{definition}

Before continuing on to examples, it is worth mentioning that a real spectral triple $(A,H,D,J)$ defines, in particular, a canonical central unital \Star-subalgebra of $A$, a fact that will be key to our discussion of real commutative and almost commutative spectral triples:

\begin{lemma}[cf.\ \cite{CCM}*{Prop.~3.1}]\label{real}
Let $(A,H,D,J)$ be a real spectral triple. Then $\tilde{A}_J := \set{a \in A \mid Ja^\ast J^\ast = a}$ defines a central unital \Star-subalgebra of $A$.
\end{lemma}

As mentioned above, by \S \ref{plymen}, we have that a compact spin $p$-manifold $X$ with spinor bundle $\cS$, charge conjugation $C$, and Dirac operator $\slashed{D}$ does give rise to a real spectral triple $(C^\infty(X),L^2(X,\cS),\slashed{D},C)$ of $KO$-dimension $p \bmod 8$, the canonical, motivating example of a real spectral triple. Indeed, we have the following well-known consequence of the reconstruction theorem for commutative spectral triples, the original form of Theorem~\ref{babyrecon}:

\begin{corollary}[Connes \cite{Con96}, Gracia-Bond{\'\i}a--V\'arilly--Figueroa \cite{GBVF}*{Thm.\ 11.2}]
Let $(A,H,D)$ be a strongly orientable $p$-dimensional commutative spectral triple such that $A^{\prime\prime}$ acts on $H$ with multiplicity $2^{\lfloor p/2\rfloor}$, so that $(A,H,D) \cong (C^\infty(X),L^2(X,\cS),D)$ for $X$ a compact spin$^\bC$ $p$-manifold and $\cS \to X$ a spinor bundle, with $D$ identified with an essentially self-adjoint Dirac-type operator on $\cS$. If, in addition, there exists an antiunitary $J$ making $(A,H,D,J)$ a real spectral triple of $KO$-dimension $p \bmod 8$, with $JaJ^\ast = a^\ast$ for $a \in A \cong C^\infty(X)$, then $X$ is spin, $\cS$ is the spinor bundle on $X$, $J$ is the charge conjugation on $\cS$, and $D = \slashed{D} + M$ for $\slashed{D}$ the Dirac operator on $X$ and $M$ a suitable symmetric bundle endomorphism on $\cS$.
\end{corollary}

Now, just as in \S \ref{plymen}, in the case that $n$ is even, we can go reversibly from the ``conventional'' $KO$-dimension $n_+$ to the ``exotic'' $KO$-dimension $n_-$ by replacing $J$ by $J\gamma$, so that we can expand Table \ref{kotable} to Table \ref{newkotable} for free. By abuse of notation and terminology, then, we shall say that $(A,H,D,\gamma,J)$ is \emph{of $KO$-dimension $n_+ \bmod 8$} if $(\eps,\epsp,\epspp)$ is given by $n_+$ in the above table, and that it is \emph{of $KO$-dimension $n_- \bmod 8$} if $(\eps,\epsp,\epspp)$ is given by $n_-$ instead. Indeed, we shall find the following definition convenient:

\begin{definition}
Let $(A,H,D,\gamma,J)$ be a real spectral triple of $KO$-dimension $n \bmod 8$ for $n$ even, and let $\beta \in \set{\pm 1}$.
\begin{itemize}
	\item If $\beta = 1$, then $J_\beta$ is the element of $\set{J,J\gamma}$ such that $(A,H,D,\gamma,J_+)$ has $KO$-dimension $n_+ \bmod 8$;
	\item If $\beta = -1$, then $J_-$ is the element of $\set{J,J\gamma}$ such that $(A,H,D,\gamma,J_-)$ has $KO$-dimension $n_- \bmod 8$.
\end{itemize}
\end{definition}

Thus, we are free to identify a real spectral triple $(A,H,D,\gamma,J)$ of even $KO$-dimension $n \bmod 8$ simultaneously with the real spectral triple $(A,H,D,\gamma,J_1)$ of $KO$-dimension $n_+ \bmod 8$ and the real spectral triple $(A,H,D,\gamma,J_{-1})$ of $KO$-dimension $n_- \bmod 8$.

\subsection{Commutative spectral triples}

Let us now consider the case of real \emph{commutative} spectral triples. We have already seen the example of the canonical real spectral triple of a compact spin manifold with fixed spin structure. However, the canonical spectral triple of a compact oriented Riemannian manifold (\emph{v.\ supra}) immediately gives rise to a canonical real spectral triple of $KO$-dimension $0 \bmod 8$, a seemingly trivial example which shall prove quite instructive indeed:

\begin{example}\label{cheapex}
Let $X$ be a compact oriented Riemannian manifold. Since the operators $d+d^\ast$ and $(-1)^{\abs{\cdot}}$ on $\wedge T^\ast_\bC X$ are simply straightforward $\bC$-linear extensions of operators on the real exterior bundle $\wedge T^\ast X$, we can realise the Hodge--de Rham spectral triple of $X$ as a real triple $(C^\infty(X),L^2(X,\wedge T^\ast_\bC X),d+d^\ast,(-1)^{\abs{\cdot}},K)$ of $KO$-dimension $0 \bmod 8$, where $K$ is the complex conjugation operator on $\wedge T^\ast_\bC X$ \emph{qua} complexification of the real vector bundle $\wedge T^\ast X$.
\end{example}

In light of this last example, we already see that a generalisation of the ``charge conjugation'' operator of Cor.~\ref{plymenspin} can be usefully defined on more general Clifford modules:

\begin{definition}
Let $X$ be a compact oriented Riemannian manifold, let $\cE \to X$ be a Clifford module, which may or may not be \ZZ-graded. Let $J$ be an antiunitary bundle automorphism on $\cE$, and let $n \in \bZ_8$. We call $(\cE,J)$ a \emph{real Clifford module of $KO$-dimension $n \bmod 8$} if $\cE$ is \ZZ-graded with \ZZ-grading $\gamma$ when $n$ is even, and $C$ satisfies the following:
\begin{enumerate}
	\item $J^2 = \eps \Id_\cE$,
	\item $Jc(\omega^\ast)J^\ast = c(\tau_{\epsp}(\omega))$ for all $\omega \in C^\infty(X,\bCl(X))$,
	\item $J\gamma = \epspp \gamma J$ if $n$ is even,
\end{enumerate}
where $(\eps,\epsp,\epspp) \in \set{\pm 1}^3$ is determined by $n \bmod 8$ according Table~\ref{kotable}.
\end{definition}

\begin{remark}
Just as in the spinor case, if $n$ is even, then we can replace $J$ with $J\gamma$ to go reversibly between the ``conventional'' $KO$-dimension $n_+$ and the ``exotic $KO$-dimension'' $n_-$.
\end{remark}

In both examples, we have a real triple of the form $(C^\infty(X),L^2(X,\cE),D,J)$, where $X$ is a compact oriented Riemannian manifold, $(\cE,J)$ is a real Clifford module, and $D$ is a Dirac-type operator on $\cE$ compatible with $J$ in the following sense:

\begin{definition}
Let $(\cE,J)$ be a real Clifford module of $KO$-dimension $n \bmod 8$ over a compact oriented Riemannian manifold $X$. Let $D$ be a Dirac-type operator on $\cE$. We shall call $D$ \emph{$J$-compatible} if $DJ = \epsp J D$.
\end{definition}

Thus, if $(\cE,J)$ is a real Clifford module of $KO$-dimension $n \bmod 8$ over $X$, and $D$ is a Dirac-type operator on $\cE$, then $(C^\infty(X),L^2(X,\cE),D,J)$ is a real spectral triple of $KO$-dimension $n \bmod 8$ if and only if $D$ is $J$-compatible---let us call such a real spectral triple a \emph{concrete real commutative spectral triple}. One can therefore ask if a Dirac-type operator on $\cE$ is necessarily $J$-compatible. As it turns out, the answer is yes, up to perturbation by a symmetric bundle endomorphism:

\begin{proposition}
Let $X$ be a compact oriented Riemannian manifold, let $(\cE,J)$ be a real Clifford module on $X$ of $KO$-dimension $n \bmod 8$, and let $D$ be a symmetric Dirac-type operator on $\cE$. Then there exists a unique symmetric bundle endomorphism $M$ on $\cE$ such that $D-M$ is a $J$-compatible Dirac-type operator, and $MJ = - \epsp JM$.
\end{proposition}

\begin{proof}
For any $f \in C^\infty(X,\bR)$, $J [D,f] J^\ast = J c(df) J^\ast = \epsp c(df) = \epsp[D,f]$, and hence $ [D - \epsp J D J^\ast,f] = 0$. Thus, $M = \tfrac{1}{2}(D - \epsp J D J^\ast)$ is a symmetric bundle endomorphism, so that $D - M = \tfrac{1}{2}(D + \epsp JDJ^\ast)$ is a symmetric Dirac-type operator; if $n$ is even, so that $\cE$ is \ZZ-graded, then $D-M$ is odd since $D$ is, and since $J^2 = \eps$ commutes with $D$ and with $M$, $J(D-M) = \epsp (D-M)J$ and $JM = -\epsp MJ$, as required.

Finally, suppose that $N$ is another symmetric bundle endomorphism on $\cE$ such that $(D-N)J = \epsp J(D-N)$ and $NJ = - \epsp JN$. Then $(D-M) - (D-N) = N- M$ both commutes and anticommutes with $J$, and thus must vanish.
\end{proof}

Finally, let us show that a Dirac-type real commutative spectral triple, in the appropriate abstract sense, necessarily arises from a real Clifford module together with compatible Dirac-type operator.

Consider a concrete real commutative spectral triple $(C^\infty(X),L^2(X,\cE),D,J)$. On the one hand, $C^\infty(X)$ is already commutative, while on the other, by Lem.~\ref{real}, $C^\infty(X)$ contains a canonical central unital \Star-subalgebra
\[
	\widetilde{C^\infty(X)}_J := \set{a \in C^\infty(X): Ja^\ast J^\ast = a};
\]
that $J$ is an anti-linear bundle endomorphism on $\cE$ is then precisely equivalent to the fact that $\widetilde{C^\infty(X)}_J = C^\infty(X)$. This, then, motivates the following:

\begin{definition}
Let $(A,H,D,J)$ be a real spectral triple. We call $(A,H,D,J)$ a \emph{real commutative spectral triple} if the following hold:
\begin{enumerate}
	\item $A = \tilde{A}_J$, or equivalently, $J a J^\ast = a^\ast$ for all $a \in A$;
	\item $(A,H,D)$ is a Dirac-type commutative spectral triple.
\end{enumerate}
\end{definition}

In particular, then, a concrete real commutative spectral triple is automatically a real commutative spectral triple in this abstract sense.

\begin{remark}
If one wants $A$ to correspond to $C^\infty(X,\bR)$ instead of $C^\infty(X) = C^\infty(X,\bC)$, then one should take $A$ to be a real (Fr\'echet) pre-\Cstar-algebra with trivial \Star-operation, in which case, condition (1) corresponds simply to commutativity of $J$ with $A$.
\end{remark}

The relevant refinement of the reconstruction theorem for commutative spectral triples is thus the claim that a real spectral triple is real commutative if and only if it is unitarily equivalent to a concrete real commutative spectral triple:

\begin{proposition}\label{rcrecon}
Let $(A,H,D,J)$ be a real commutative spectral triple of $KO$-dimension $n \bmod 8$ and metric dimension $p$. Then there exist a compact oriented Riemannian $p$-manifold $X$ and a self-adjoint Clifford module $\cE \to X$ such that $(A,H,D,J) \cong (C^\infty(X),L^2(X,\cE),D,J)$, where $D$, viewed as an operator on $L^2(X,\cE)$, is an essentially self-adjoint Dirac-type operator on $\cE$, and where $J$, viewed as an operator on $L^2(X,\cE)$, makes $(\cE,J)$ a real Clifford module of $KO$-dimension $n \bmod 8$ such that $D$ is $J$-compatible.
\end{proposition}

\begin{proof}
Suppose that $(A,H,D,J)$ is a real commutative spectral triple of $KO$-dimension $n \bmod 8$ and metric dimension $p$. In particular, $(A,H,D)$ is a Dirac-type commutative spectral triple of metric dimension $p$, so that by Thm.~\ref{rierecon}, there exist a compact oriented Riemannian $p$-manifold $X$ and a Hermitian vector bundle $\cE \to X$ such that $(A,H,D) \cong (C^\infty(X), L^2(X,\cE),D)$, where $D$, viewed as an operator on $L^2(X,\cE)$, defines an essentially self-adjoint Dirac-type operator on $\cE$. In particular, then, $D$ makes $\cE$ into a Clifford module, so that it suffices to prove that $J$, viewed as an operator on $L^2(X,\cE)$, is an anitunitary bundle automorphism on $\cE$ making $(\cE,J)$ a real Clifford module of $KO$-dimension $n \bmod 8$.

First, since $JaJ^\ast = a^\ast$ for all $a \in C^\infty(X)$, $J$ can be viewed as a unitary $C^\infty(X)$-linear morphism $C^\infty(X,\cE) \to \overline{C^\infty(X,\cE)} = C^\infty(X,\overline{\cE})$, where $\overline{C^\infty(X,\cE)}$ is the conjugate $C^\infty(X)$-module to $C^\infty(X,\cE)$, and $\overline{\cE}$ is the conjugate bundle to $\cE$. Hence, $J$ defines a unitary bundle isomorphism $\cE \cong \overline{\cE}$, that is, an antiunitary bundle automorphism on $\cE$. The rest then follows from the fact that $(C^\infty(X),L^2(X,\cE),D,J)$ is a real spectral triple of $KO$-dimension $n \bmod 8$; in particular, since $DJ = \epsp JD$, $Jc(df)J^\ast = J[D,f]J^\ast = \epsp[D,f] = \epsp c(df)$ for $f \in C^\infty(X,\bR)$, as required.
\end{proof}

\begin{remark}
One may ask which $KO$-dimensions are possible for real commutative spectral triples over a given compact oriented Riemannian manifold $X$. We shall soon see how to use the spectral triple of Ex.~\ref{cheapex} to construct real commutative spectral triples over $X$ of any $KO$-dimension.
\end{remark}

\subsection{Almost-commutative spectral triples}

At last, let us consider real structures on almost-commutative spectral triples, both concrete and abstract. To see what a real structure looks like on a concrete almost-commutative spectral triple, it suffices to consider the traditional Cartesian product construction. Let us therefore recall the construction of a product of real spectral triples as formulated by D\k{a}browski--Dossena, after Vanhecke:

\begin{theorem}[D\k{a}browski--Dossena~\cite{DaDo}*{\S 4}, cf.\ Vanhecke~\cite{Van99}]\label{dd}
For $i=1,2$, let $X_i = (A_i,H_i,D_i,J_i)$ be a real spectral triple of $KO$-dimension $n_i \bmod 8$. Then $X_1 \times X_2$ can made into a real spectral triple of $KO$-dimension $n_1 + n_2 \bmod 8$ with real structure $J$ defined as follows:
\begin{enumerate}
	\item If $n_1$ and $n_2$ are both even, then
	\[
		J_\pm := (J_1)_{\pm \epspp(n_1)} \otimes (J_2)_\pm;
	\]
	\item If $n_1$ is even and $n_2$ is odd, then
	\[
		J := (J_1)_{\epsp(n_1 + n_2)} \otimes J_2;
	\]
	\item If $n_1$ is odd and $n_2$ is even, then
	\[
		J := J_1 \otimes (J_2)_{\epsp(n_1+n_2)};
	\]
	\item If $n_1$ and $n_2$ are both odd, then
	\[
		J_\pm := J_1 \otimes J_2 \otimes M_\pm K,
	\]
	for $K$ the complex conjugation on $\bC^2$ and $(M_+,M_-)$ chosen as follows, with rows indexed by $n_1$ and columns indexed by $n_2$:
	\begin{equation}\label{oddodd}
		\begin{array}{ccccc}
			\toprule
			& 1 & 3 & 5 & 7\\
			\midrule
			1 & (i\sigma_2,\sigma_1) & (\sigma_3,\sigma_0) & (i\sigma_2,\sigma_1) & (\sigma_3,\sigma_0)\\
			3 & (\sigma_0,\sigma_3) & (\sigma_1,i\sigma_2) & (\sigma_0,\sigma_3) & (\sigma_1,i\sigma_2)\\
			5 & (i\sigma_2,\sigma_1) & (\sigma_3,\sigma_0) & (i\sigma_2,\sigma_1) & (\sigma_3,\sigma_0)\\
			7 & (\sigma_0,\sigma_3) & (\sigma_1,i\sigma_2) & (\sigma_0,\sigma_3) & (\sigma_1,i\sigma_2)\\
			\bottomrule
		\end{array}
	\end{equation}
\end{enumerate}
\end{theorem}

Now, let $X:=(C^\infty(X),L^2(X,\cS),\slashed{D},C)$ be the canonical spectral triple of a compact spin $n_1$-manifold $X$ with fixed spin structure, and let $F  = (A_F,H_F,D_F,J_F)$ be a finite real spectral triple of $KO$-dimension $n_2 \bmod 8$. Ignoring the real structures, one has that $X \times F$ takes the form $(C^\infty(X,\cA),L^2(X,\cE),D)$, where, in particular, $\cA := X \times A_F$ is a bundle of algebras, $\cE$, formed from $\cS$ and $H_F$, is a Clifford $\cA$-module with Clifford action defined by $c(df) := [D,f]$ for $f \in C^\infty(X)$, and $\cA$-module structure induced by the representation of $A_F$ on $H_F$, and $D$, formed from $\slashed{D}$ and $D_F$, is a Dirac-type operator on the Clifford module $\cE$. Now, since $X$ and $F$ are real spectral triples, by the above theorem, $X \times F$ is a real spectral triple with real structure $J$. Taking this into account, we see that $\cE$ is a Clifford $\cA \otimes \cA^o$-module with $\cA \otimes \cA^o$-module structure defined by
\[
	(a \otimes b^o)\xi := aJb^\ast J^\ast\xi, \quad a, b \in C^\infty(X,\cA), \; \xi \in C^\infty(X,\cE),
\]
that $D$ therefore satisfies the addition constraint that
\[
	[[D,a],b^o] = 0, \quad a, b \in C^\infty(X,\cA),
\]
and hence that $J$ is an antiunitary bundle endomorphism on $\cE$ satisfying
\[
	Jc(\omega^\ast \otimes a^\ast \otimes (b^\ast)^o) J^\ast = c(\tau_{\epsp}(\omega) \otimes b \otimes a^o), \quad \omega \in C^\infty(X,\bCl(X)), \; a, b \in C^\infty(X,\cA).
\]
Thus, the additional structure provided by the real structure $J$ is encoded in the following definition:

\begin{definition}
Let $X$ be a compact oriented Riemannian manifold, let $\cA \to X$ be a bundle of algebras, and let $\cE \to X$ be a Clifford $\cA \otimes \cA^o$-module, which may or may not be \ZZ-graded. Let $J$ be an antiunitary bundle automorphism on $\cE$, and let $n \in \bZ_8$. We call $(\cE,J)$ a \emph{real Clifford $\cA$-bimodule of $KO$-dimension $n \bmod 8$} if $\cE$ is \ZZ-graded with \ZZ-grading $\gamma$ when $n$ is even, and $J$ satisfies the following:
\begin{enumerate}
	\item $J^2 = \eps \Id_\cE$,
	\item for all $\omega \in C^\infty(X,\bCl(X))$ and $a$, $b \in C^\infty(X,\cA)$,
	\[J(c(\omega^\ast \otimes a^\ast \otimes (b^\ast)^o)J^\ast = c(\tau_{\epsp}(\omega) \otimes b \otimes a^o),\]
	\item $J\gamma = \epspp \gamma J$ if $n$ is even,
\end{enumerate}
where $(\eps,\epsp,\epspp) \in \set{\pm 1}^3$ is determined by $n \bmod 8$ according Table~\ref{kotable}.
\end{definition}

\begin{remark}
Once more, just as before, if $n$ is even, then we can replace $J$ with $J\gamma$ to go reversibly between the ``conventional'' $KO$-dimension $n_+$ and the ``exotic $KO$-dimension'' $n_-$.
\end{remark}

\begin{remark}
Condition (2) in the above definition can be viewed as specifying the compatibility of $J$ with the Clifford $\cA$-bimodule structure on $\cE$, for $J$, \emph{a priori}, defines a $\bC$-linear anti-involution $T \mapsto J T^\ast J^\ast$ on $\End(\cE)$.
\end{remark}

The compatibility of the Dirac-type operator $D$ with the real Clifford $\cA$-bimodule $(\cE,J)$ is then encoded in the following definition:

\begin{definition}
Let $X$ be a compact oriented Riemannian manifold, let $\cA \to X$ be a bundle of algebras, and let $(\cE,J)$ be a real Clifford $\cA$-bimodule of $KO$-dimension $n \bmod 8$. Let $D$ be a Dirac-type operator on $\cE$. We shall call $D$ \emph{$(\cA,J)$-compatible} if it is $J$-compatible and
\[ [[D,a],b^o] = 0, \quad a, b \in C^\infty(X,\cA).\]
\end{definition}

Thus, if $(\cE,J)$ is a real Clifford $\cA$-bimodule of $KO$-dimension $n \bmod 8$ over $X$, and $D$ is a Dirac-type operator on $\cE$, then $(C^\infty(X,\cA),L^2(X,\cE),D,J)$ defines a real spectral triple of $KO$-dimension $n \bmod 8$ if and only if $D$ is $(\cA,J)$-compatible. Indeed, one can therefore give the following definition, generalising the example of the product of a spin manifold with a finite real spectral triple:

\begin{definition}
A \emph{concrete real almost-commutative spectral triple} is a real spectral triple of the form $(C^\infty(X,\cA),L^2(X,\cE),D,J)$, where $X$ is a compact oriented Riemannian manifold, $\cA \to X$ is a bundle of algebras, $(\cE,J)$ is a real Clifford $\cA$-bimodule, and $D$ is a $(\cA,J)$-compatible Dirac-type operator on $\cE$.
\end{definition}

In fact, in light of this more general definition, we can take any concrete real commutative spectral triple $V := (C^\infty(X),L^2(X,\cV),D_\cV,J_\cV)$ of $KO$-dimension $n_1 \bmod 8$ and any finite real spectral triple $F$ of $KO$-dimension $n_2 \bmod 8$ to form the concrete real almost-commutative spectral triple $V \times F$ of $KO$-dimension $n_1 + n_2 \bmod 8$. Applying this to Ex.~\ref{cheapex}, we immediately obtain the following

\begin{proposition}
Let $X$ be a compact oriented Riemannian manifold. Then for any $n \in \bZ_8$ there exists a concrete real almost-commutative spectral triple of $KO$-dimension $n \bmod 8$, namely, $(C^\infty(X),L^2(X,\wedge T^\ast_\bC X),d+d^\ast,(-1)^{\abs{\cdot}},K) \times F$ for any finite real spectral triple $F$ of $KO$-dimension $n \bmod 8$, and hence a concrete real commutative spectral triple of $KO$-dimension $n \bmod 8$.
\end{proposition}

Given a real Clifford $\cA$-bimodule $(\cE,J)$, one can again ask, just as in the commutative case, if a Dirac-type operator on $\cE$ compatible with its Clifford action and the bimodule structure is necessarily $(\cA,J)$-compatible. Once more, the answer is yes, up to perturbation by a symmetric bundle endomorphism:

\begin{proposition}
Let $X$ be a compact oriented Riemannian manifold, let $(\cE,J)$ be a real Clifford module on $X$ of $KO$-dimension $n \bmod 8$, and let $D$ be a Dirac-type operator on $\cE$ satisfying
\[
	[[D,a],b^o] = 0, \quad a, b \in C^\infty(X,\cA).
\]
Then there exists a unique symmetric bundle endomorphism $M$ on $\cE$ such that $D-M$ is an $(\cA,J)$-compatible Dirac-type operator, and  $MJ = - \epsp JM$.
\end{proposition}

Now, let us derive the corresponding abstract definition of real almost-com\-mu\-ta\-tive spectral triple, so that we can get the appropriate refinement of the reconstruction theorem for almost-commutative spectral triples.

Let $(A,H,D,J) := (C^\infty(X,\cA),L^2(X,\cE),D,J)$ be a concrete real almost-commu\-tative spectral triple. On the one hand, $(A,H,D)$ is, in particular, an abstract almost-commutative spectral triple with base (viz, distinguished central unital \Star-subalgebra of $A$) $B := C^\infty(X)$. On the other hand, by Lem.~\ref{real}, $A$ contains a canonical central unital \Star-subalgebra $\tilde{A}_J$, which moreover contains $B$ precisely because $J$ is, in particular, an antilinear bundle endomorphism of $\cE$. It has already been observed in specific examples (e.g., the noncommutative-geometric Standard Model~\cite{CCM}*{Lemma 3.2})that $B$ and $\tilde{A}_J$ are, in fact, equal---as it turns out, this is a completely general phenomenon.

\begin{proposition}\label{realbase}
Let $(C^\infty(X,\cA),L^2(X,\cE),D,J)$ be a concrete real almost-com\-mutative spectral triple. Then $\widetilde{C^\infty(X,\cA)}_J = C^\infty(X)1_\cA.$
\end{proposition}

\begin{remark}
If $\cA$ is a real bundle of algebras, then one should replace $C^\infty(X)$ with $C^\infty(X,\bR)$ in the above statement.
\end{remark}

This result is an immediate corollary of the following algebraic observation, applied pointwise:

\begin{lemma}
Let $A_F$ be a finite \Cstar-algebra over $\bK = \bR$ or $\bC$. Then
\[
	\set{a \in A_F \mid a \otimes 1 = 1 \otimes a \in A_F \otimes_\bK A_F } = \bK 1_{A_F}.
\]
\end{lemma}

\begin{proof}
By Wedderburn's theorem for finite-dimensional \Cstar-algebras, write
\[
	A_F = \oplus_{k=1}^N M_k (\bK_k)
\]
where $\bK_k \in \set{\bR,\bC,\bH}$ if $\bK = \bR$, and $\bK_k = \bC$ if $\bK = \bC$. By construction, then,
\[
	\set{a \in A_F \mid a \otimes 1 = 1 \otimes a \in A_F \otimes_\bK A_F} \subset Z(A_F) \cong \bigoplus_{k=1}^N \bK^\prime_k,
\]
where $\bK^\prime_k := \bC$ if $\bK_k = \bC$, and $\bK^\prime_k := \bR$ otherwise, so that
\[
	Z(A_F) \otimes_\bK Z(A_F) = \bigoplus_{k,l=1}^N \bK^\prime_k \otimes_\bK \bK^\prime_l.
\]
Now, let $a \in Z(A_F)$, which we identify with $(\lambda_k)_{k=1}^N \in \oplus_{k=1}^N \bK^\prime_k$. Then
\[
	a \otimes 1 - 1 \otimes a = (\lambda_k \otimes 1 - 1 \otimes \lambda_l)_{k,l = 1}^N,
\]
so that $a \otimes 1 = 1 \otimes a$ if and only if $\lambda_k \otimes 1 = 1 \otimes \lambda_k$ for all $1 \leq k,l \leq N$. We have two cases. First, suppose that $\bK = \bC$. It therefore follows that $a \otimes 1 = 1 \otimes a$ if and only if $\lambda_k = \lambda_l$ for all $k$, $l$, if and only if $a \in \bC 1_{A_F}$. Now, suppose that $\bK = \bR$. Then, similarly, $a \otimes 1 = 1 \otimes a$ if and only if $\lambda_k = \lambda_l \in \bR$ for all $k$, $l$, if and only if $a \in \bR 1_{A_F}$.
\end{proof}

Our observations motivate the following definition:

\begin{definition}
Let $(A,H,D,J)$ be a real spectral triple. We call $(A,H,D,J)$ a \emph{real almost-commutative spectral triple} if $(A,H,D)$ is an almost-commutative spectral triple with base $\tilde{A}_J$.
\end{definition}

We have just seen that every concrete real almost-commutative spectral triple is a real almost-commutative spectral triple; the reconstruction theorem for almost-commutative spectral triples readily implies the converse.

\begin{theorem}
Let $(A,H,D,J)$ be a real almost-commutative spectral triple of $KO$-dimension $n \bmod 8$ and metric dimension $p$. Then there exist a compact oriented Riemannian $p$-manifold $X$, a bundle of algebras $\cA \to X$, and a Clifford $\cA$-bimodule $\cE \to X$ such that $\tilde{A}_J \cong C^\infty(X)$ and
\[
	(A,H,D,J) \cong (C^\infty(X,\cA), L^2(X,\cE),D,J),
\]
where, viewing $D$ and $J$ as operators on $\cE$, $(\cE,J)$ is a real Clifford $\cA$-bimodule of $KO$-dimension $n \bmod 8$, and $D$ is a $(\cA,J)$-compatible essentially self-adjoint Dirac-type operator on $\cE$.
\end{theorem}

\begin{remark}
If $A$ is taken to be a real (Fr\'echet) pre-\Cstar-algebra, and not complex, then one finds instead that $\tilde{A}_J \cong C^\infty(X,\bR)$.
\end{remark}

\begin{proof}
First, by Thm.~\ref{acrecon}, there exist a compact oriented Riemannian $p$-manifold $X$, a bundle of algebras $\cA \to X$, and a Clifford $\cA$-module $\cE \to X$ such that $\tilde{A}_J \cong C^\infty(X)$ and
$(A,H,D) \cong (C^\infty(X,\cA),L^2(X,\cE),D)$, where $D$, viewed as an operator on $\cE$, is an essentially self-adjoint Dirac-type operator. Viewing $J$ as an operator on $L^2(X,\cE)$, in light of Prop.~\ref{rcrecon}, we therefore have that $(C^\infty(X),L^2(X,\cE),D,J)$ is a real Dirac-type commutative spectral triple of $KO$-dimension $n \bmod 8$, and hence that $(\cE,J)$ is a real Clifford module of $KO$-dimension $n \bmod 8$. Finally, the conditions for a real spectral triple imply that $(\cE,J)$ is a real Clifford $\cA$-bimodule, with $\cA \otimes \cA^o$-module structure given by
\[
	(a \otimes b^o)\xi := aJb^\ast J^\ast \xi, \quad a,b \in C^\infty(X,\cA), \; \xi \in C^\infty(X,\cE),
\]
and that $D$ is $(\cA,J)$-compatible, as required.
\end{proof}

\section{Twistings}

We have already seen how to generalise the conventional definition of real almost-commutative spectral triple into a form suited to a reconstruction theorem. For physical applications, however, it is useful to consider a more conservative generalisation, where we take the product of a concrete real commutative spectral triple not with a single finite real spectral triple, but with a family of such spectral triples equipped with suitable connection:

\begin{definition}
Let $X$ be a compact oriented Riemannian manifold. A \emph{real family of $KO$-dimension $n \bmod 8$} over $X$ is a quintuple of the form $(\cA,\cF,\nabla^\cF,D_\cF,J_\cF)$, where:
\begin{enumerate}
	\item $\cA \to X$ is a bundle of algebras;
	\item $\cF \to X$ is an $\cA \otimes \cA^o$-module endowed, if $n$ is even, with a $\ZZ$-grading $\gamma_\cF$  commuting with all sections of $\cA \otimes \cA^o$;
	\item $\nabla^\cF$ is a self-adjoint connection on $\cF$, odd if $n$ is even, such that the induced connection on $\End(\cF)$ restricts to connections on $\cA$ and on $\cA^o$;
	\item $D_\cF$ is a symmetric bundle endomorphism on $\cF$, odd if $n$ is even, satisfying
	\[
		[[D_\cF,a],b^o] = 0, \quad a,b \in C^\infty(X,\cA),
	\]
	and which anticommutes with $\gamma_\cF$ if $n$ is even;
	\item $J_\cF$ is an antiunitary bundle endomorphism on $\cF$ such that
		\begin{enumerate}
		\item $J_\cF^2 = \eps \Id_\cF$,
		\item $D_\cF J_\cF = \epsp J_\cF D_\cF$, $\nabla^\cF \circ J_\cF = J_F \circ \nabla^\cF$, and $J_\cF a^\ast J_\cF^\ast = a^o$ for all $a \in C^\infty(X,\cA)$,
		\item $\gamma_\cF J_\cF = \epspp J_\cF \gamma_\cF$, if $n$ is even,
		\end{enumerate}
		where $(\eps,\epsp,\epspp) \in \set{\pm 1}^3$ depend on $n \bmod 8$ according to Table~\ref{kotable} (or, equivalently, according to Table~\ref{newkotable}).
\end{enumerate}
\end{definition}

\begin{remark}For each $x \in X$, $(\cA_x,\cF_x,(D_\cF)_x,(J_\cF)_x)$ is a finite real spectral triple of $KO$-dimension $n \bmod 8$, and the real family can be viewed as a family $(\cA,\cF,D_\cF,J_\cF)$ of finite real spectral triples over $X$ together with Bismut superconnection $D_\cF + \nabla^\cF$.
\end{remark}

If $F = (A_F,H_F,D_F,J_F)$ is a single finite real spectral triple of $KO$-dimension $n \bmod 8$, and $X$ is a compact oriented Riemannian manifold, then for $\cA := X \times A_F$ we can define a real family of $KO$-dimension $n \bmod 8$ by
\[
	(\cA,\cF,\nabla^\cF,D_\cF,J_\cF) := (X \times A_F,X \times H_F,d, \Id_X \times D_F, \Id_X \times J_F).
\]
More generally, let $G$ be a compact Lie group acting on $F$, in the sense that there exists a unitary representation $U : G \to U(H_F)$ such that for all $g \in G$,
\[
	U(g)A_F U(g)^\ast \subset A_F, \quad [U(g),D_F] = 0, \quad [U(g),J_F] = 0,
\]
with each $U(g)$ even if $n$ is even, and let $\cP \to X$ be a principal $G$-bundle with connection $\nabla^\cP$. Then we can define a real family of $KO$-dimension $n \bmod 8$ by
\[
	(\cA, \cF,\nabla^\cF,D_\cF,J_\cF) := (\cP \times_G A_F, \cP \times_G H_F, \nabla^{\cP \times_G H_F},\Id_\cP \times D_F,\Id_\cP \times J_F),
\]
with \ZZ-grading, if $n$ is even, given by $\gamma_\cF := \Id_\cP \times \gamma_F$ for $\gamma_F$ the $\ZZ$-grading of $F$.

In light of Def.~\ref{products} and Thm.~\ref{dd}, one therefore defines the twisting of a concrete real commutative spectral triple by a real family as follows:

\begin{definition}
Let $(C^\infty(X),L^2(X,\cV),D_\cV,J_\cV)$ be a concrete real commutative spectral triple of $KO$-dimension $m \bmod 8$ and metric dimension $p$, with \ZZ-grading $\gamma_\cV$ if $m$ is even and let $(\cA,\cF,\nabla^\cF,D_\cF,J_\cF)$ be a real family of $KO$-dimension $n \bmod 8$, with \ZZ-grading $\gamma_\cF$ is $n$ is even. Then the \emph{twisting} of $(C^\infty(X),L^2(X,\cV),D_\cV,J_\cV)$ by $(\cA,\cF,\nabla^\cF,D_\cF,J_\cF)$ is the concrete real almost-commutative spectral triple
\[
	(\cA,\cF,\nabla^\cF,D_\cF,J_\cF) \times (C^\infty(X),L^2(X,\cV),D_\cV,J_\cV) := (C^\infty(X,\cA),L^2(X,\cE),D,J),
\]
of $KO$-dimension $m+n \bmod 8$ and metric dimension $p$, where $\cE$, $D$ and $J$ are defined as follows:
\begin{enumerate}
	\item if $m$ and $n$ are both even, then
	\[
		\cE := \cV \otimes \cF, \quad D :=  D_\cV \otimes_{\nabla^\cF} 1 + \gamma_\cV \otimes D_\cF, \quad J := (J_\cV)_{\epspp(m)} \otimes J_\cF,
	\]
	with \ZZ-grading $\gamma := \gamma_\cV \otimes \gamma_\cF$;
	\item if $m$ is even and $n$ is odd, then
	\[
		\cE := \cV \otimes \cF, \quad D :=  D_\cV \otimes_{\nabla^\cF} 1 + \gamma_\cV \otimes D_\cF, \quad J := (J_\cV)_{\epsp(m+n)} \otimes J_\cF;
	\]
	\item if $m$ is odd and $n$ is even, then
	\[
		\cE := \cV \otimes \cF, \quad D :=  D_\cV \otimes_{\nabla^\cF} \gamma_\cF + 1 \otimes D_\cF, \quad J := J_\cV \otimes (J_\cF)_{\epsp(m+n)};
	\]
	\item if $m$ and $n$ are both odd, then
	\[
		\cE := \cV \otimes \cF \otimes \bC^2, \quad D := D_\cV \otimes_{\nabla^\cF} 1 \otimes \sigma_1 + 1 \otimes D_\cF \otimes \sigma_2, \quad J := J_\cV \otimes J_\cF \otimes M K,
	\]
	with \ZZ-grading $\gamma := 1 \otimes 1 \otimes \sigma_3$, where $K$ is the complex conjugation on $\bC^2$ and $M := M_+$ is given by Table 3.3.
\end{enumerate}
In the expressions above, if $T = 1$ or $\gamma_\cF$, then $D_\cV \otimes_{\nabla^\cF} T$ is defined locally by
\[
	(D_\cV \otimes_{\nabla^\cF} T)(\eta \otimes \xi) := (D_\cV \eta) \otimes \xi + \sum_i (c(e^i)\eta) \otimes \nabla_{e_i}^\cF (T \xi), \quad \eta \in C^\infty(X,\cV), \; \xi \in C^\infty(X,\cF),
\]
where $\set{e_i}$ is a local vielbein on $TX$.
\end{definition}

The essential point in checking that this definition makes sense is checking that one does indeed get a concrete almost-commutative spectral triple.

\begin{remark}
Let $(C^\infty(X),L^2(X,\cV),D_\cV,J_\cV)$ be a concrete real commutative spectral triple and let $(\cA,\cF,\nabla^\cF,D_\cF,J_\cF)$ be a real family; for simplicity, suppose that both are of even $KO$-dimension. Then $(C(X,\cF),D_\cF,\nabla^\cF)$ can be viewed as an unbounded $(C(X,\cA),C(X))$-bimodule in the sense of Mesland, and the twisting
\[
	(\cA,\cF,\nabla^\cF,D_\cF,J_\cF) \times (C^\infty(X),L^2(X,\cV),D_\cV,J_\cV) =: (C^\infty(X,\cA),L^2(X,\cE),D,J)
\]
can be viewed as an unbounded Kasparov product, that is,
\[
	(L^2(X,\cE),D) \cong (C(X,\cF),D_\cF,\nabla^\cF) \times (L^2(X,\cV),D_\cV),
\]
and hence $(\cF,\nabla^\cF,D_\cF)$ defines a morphism
\[
	(C(X,\cF),D_\cF,\nabla^\cF) : (C^\infty(X,\cA),L^2(X,\cE),D) \to (C^\infty(X),L^2(X,\cV),D_\cV)
\]
in Mesland's category of spectral triples.
\end{remark}

Finally, let us record the consequences of this construction for the structure of inner fluctuations of the metric:

\begin{proposition}
Let $(C^\infty(X),L^2(X,\cV),D_\cV,J_\cV)$ be a concrete real commutative spectral triple of $KO$-dimension $m \bmod 8$ and let $(\cA \cF,\nabla^\cF,D_\cF,J_\cF)$ be a real family of $KO$-dimension $n \bmod 8$. Let
\[
	(C^\infty(X,\cA),L^2(X,\cE),D,J) := (\cA,\cF,\nabla^\cF,D_\cF,J_\cF) \times (C^\infty(X),L^2(X,\cV),D_\cV,J_\cV),
\]
which is a concrete real almost-commutative spectral triple of $KO$-dimension $m+n \bmod 8$. Let $\bA \in C^\infty(X,\End(\cE))$ be an inner fluctuation of the metric on the twisting, i.e., symmetric and of the form $\bA = \sum_i a_i[D,b_i]$ for $a_i$, $b_i \in C^\infty(X,\cA)$, and let
\begin{align*}
	\omega_\bA &:= \sum_i \left( a_i \wedge (\nabla^\cF b_i) - (\nabla^\cF b_i^o) \wedge a_i^o\right) \in \Omega^1(X,\cA \otimes \cA^o),\\
	\Phi_\bA &:= \sum_i \left( a_i [D_\cF,b_i] - [D_\cF,b_i^o]a_i^o \right) \in C^\infty(X,\End(\cF)).
\end{align*}
Then $(\cA,\cF,\nabla^\cF + \omega_\bA,D_\cF + \Phi_\bA,J_\cF)$ is a real family of $KO$-dimension $n \bmod 8$ such that
\begin{multline*}
	(C^\infty(X,\cA),L^2(X,\cE),D + \bA + \epsp J \bA J^\ast,J) \\
	= (\cA,\cF,\nabla^\cF + \omega_\bA,D_\cF + \Phi_\bA,J_\cF) \times (C^\infty(X),L^2(X,\cV),D_\cV,J_\cV).
\end{multline*}
\end{proposition}

Thus, for a real almost-commutative spectral triple formed by the twisting of a concrete real commutative spectral triple by a real family, inner fluctuations of the metric are are effected at the level of the real family, so that the concrete real commutative spectral triple may be viewed strictly as background data. For a general, $KK$-theoretic  discussion of this kind of phenomenon, see~\cite{BMvS}.


\begin{bibdiv}
\begin{biblist}
\bib{BvS10}{article}{
	author={Boeijink, J.},
	author={van Suijlekom, W. D.},
	title={The noncommutative geometry of Yang-Mills fields},
	journal={J.\ Geom.\ Phys.},
	volume={61},
	date={2011},
	number={6},
	pages={1122-1134},
}
\bib{BMvS}{article}{
	author={Brain, S.},
	author={Mesland, B.},
	author={van Suijlekom, W. D.},
	title={Gauge theory and the unbounded Kasparov product},
	note={(in preparation)},
}
\bib{Ca12}{article}{
   author={{\'C}a{\'c}i{\'c}, B.},
   title={A reconstruction theorem for almost-commutative spectral triples},
   journal={Lett. Math. Phys.},
   volume={100},
   date={2012},
   number={2},
   pages={181--202},
}
\bib{Con95b}{article}{
   author={Connes, A.},
   title={Noncommutative geometry and reality},
   journal={J. Math. Phys.},
   volume={36},
   date={1995},
   number={11},
   pages={6194--6231},
}
\bib{Con96}{article}{
   author={Connes, A.},
   title={Gravity coupled with matter and the foundation of non-commutative
   geometry},
   journal={Comm. Math. Phys.},
   volume={182},
   date={1996},
   number={1},
   pages={155--176},
}
\bib{Con08}{article}{
	author={Connes, A.},
	title={On the spectral characterization of manifolds},
	date={2008},
	eprint={arXiv:0810.2088v1 [math.OA]}
}
\bib{CCM}{article}{
   author={Chamseddine, A. H.},
   author={Connes, A.},
   author={Marcolli, M.},
   title={Gravity and the standard model with neutrino mixing},
   journal={Adv. Theor. Math. Phys.},
   volume={11},
   date={2007},
   number={6},
   pages={991--1089},
}
\bib{DaDo}{article}{
   author={D{\polhk{a}}browski, L.},
   author={Dossena, G.},
   title={Product of real spectral triples},
   journal={Int. J. Geom. Methods Mod. Phys.},
   volume={8},
   date={2011},
   number={8},
   pages={1833--1848},
}
\bib{GBVF}{book}{
   author={Gracia-Bond{\'{\i}}a, J. M.},
   author={V{\'a}rilly, J. C.},
   author={Figueroa, H.},
   title={Elements of noncommutative geometry},
   series={Birkh\"auser Advanced Texts: Basler Lehrb\"ucher},
   publisher={Birkh\"auser Boston Inc.},
   place={Boston, MA},
   date={2001},
}
\bib{LRV}{article}{
   author={Lord, S.},
   author={Rennie, A.},
   author={V{\'a}rilly, J. C.},
   title={Riemannian manifolds in noncommutative geometry},
   journal={J. Geom. Phys.},
   volume={62},
   date={2012},
   number={7},
   pages={1611--1638},
}
\bib{Me09b}{article}{
	author={Mesland, B.},
	title={Unbounded bivariant $K$-theory and correspondences in noncommutative geometry},
	journal={J. Reine Angew. Math.},
	note={(in press)},
}
\bib{Otgo09}{thesis}{
   author={Otgonbayar, U.},
   title={Local index theorem in noncommutative geometry},
   type={Ph.D.\ dissertation},
   organization={The Pennsylvania State University},
   date={2009},
}
\bib{Pl86}{article}{
   author={Plymen, R. J.},
   title={Strong Morita equivalence, spinors and symplectic spinors},
   journal={J. Operator Theory},
   volume={16},
   date={1986},
   number={2},
   pages={305--324},
}
\bib{RV06}{article}{
	author={Rennie, A.},
	author={V\'arilly, J. C.},
	title={Reconstruction of manifolds in noncommutative geometry},
	date={2006},
	eprint={arXiv:math/0610418v4 [math.OA]},
}
\bib{Van99}{article}{
   author={Vanhecke, F. J.},
   title={On the product of real spectral triples},
   journal={Lett. Math. Phys.},
   volume={50},
   date={1999},
   number={2},
   pages={157--162},
}
\end{biblist}
\end{bibdiv}
\end{document}